\documentclass[oneside,english]{amsart}
\usepackage[foot]{amsaddr}
\usepackage[T1]{fontenc}
\usepackage[utf8]{luainputenc}
\usepackage{geometry}
\usepackage{amsthm}
\usepackage{amssymb}
\usepackage{graphicx}
\usepackage{bigints}
\usepackage{color}
\usepackage{comment}

\makeatletter

\numberwithin{equation}{section}
\numberwithin{figure}{section}
\theoremstyle{plain}
\newtheorem{thm}{\protect\theoremname}
  \theoremstyle{plain}
  \newtheorem{prop}[thm]{\protect\propositionname}
  \theoremstyle{remark}
  \newtheorem{rem}[thm]{\protect\remarkname}
  \theoremstyle{remark}
  
  \theoremstyle{plain}
  \newtheorem{lem}[thm]{\protect\lemmaname}
  \theoremstyle{plain}

\makeatother

\usepackage{babel}
  \providecommand{\lemmaname}{Lemma}
  \providecommand{\propositionname}{Proposition}
  \providecommand{\remarkname}{Remark}
  \providecommand{\examplename}{Example}
\providecommand{\theoremname}{Theorem}
\providecommand{\definitionname}{Definition}

\newcommand{\be}{\begin{equation}}
\newcommand{\ee}{\end{equation}}

\newcommand{\bea}{\begin{eqnarray}}
\newcommand{\eea}{\end{eqnarray}}
\newcommand{\beaa}{\begin{eqnarray*}}
\newcommand{\eeaa}{\end{eqnarray*}}
\newcommand{\bes}{\begin{subequations}}
\newcommand{\ees}{\end{subequations}}

\newcommand{\EE}{{\mathbb E}}
\newcommand{\RR}{{\mathbb R}}

\newcommand{\NN}{{\mathbb N}}

\newcommand{\PP}{{\mathbb P}}

\newcommand{\cd}{{\mathcal D}}
\newcommand{\cf}{{\mathcal F}}

\newcommand{\cl}{{\mathcal L}}

\newcommand{\ind}{{\bf 1}}

\newcommand{\eps}{\varepsilon}

\DeclareMathOperator{\esssup}{ess\,sup}

\begin{document}
\title{Inversion of Convex Ordering: Local Volatility Does Not Maximize the Price of VIX Futures}

\author{Beatrice Acciaio and Julien Guyon}
\email{b.acciaio@lse.ac.uk}
\address{Department of Statistics, London School of Economics and Political Science}
\email{jguyon2@bloomberg.net}
\address{Quantitative Research, Bloomberg L.P.}
\date{\today}

\begin{abstract}
It has often been stated that, within the class of continuous stochastic volatility models calibrated to vanillas, the price of a VIX future is maximized by the Dupire local volatility model. In this article we prove that this statement is incorrect: we build a continuous stochastic volatility model in which a VIX future is \emph{strictly more expensive} than in its associated local volatility model. More generally, in this model, strictly convex payoffs on a squared VIX are strictly cheaper than in the associated local volatility model. This corresponds to an \emph{inversion of convex ordering} between local and stochastic variances, when moving from instantaneous variances to squared VIX, as convex payoffs on instantaneous variances are always cheaper in the local volatility model. We thus prove that this inversion of convex ordering, which is observed in the SPX market for short VIX maturities, can be produced by a continuous stochastic volatility model. We also prove that the model can be extended so that, as suggested by market data, the convex ordering is preserved for long maturities.
\end{abstract}

\keywords{VIX, VIX futures, stochastic volatility, local volatility, convex order, inversion of convex ordering}

\maketitle

\section{Introduction}

For simplicity, let us assume zero interest rates, repos, and dividends. Let $\cf_t$ denote the market information available up to time $t$. We consider continuous stochastic volatility models on the SPX index of the form
\bea
\frac{dS_t}{S_t} = \sigma_t \, dW_t, \qquad S_0 = s_0, \label{eq:continuous_model_on_S}
\eea
where $W=(W_t)_{t\geq0}$ denotes a standard one-dimensional $(\cf_t)$-Brownian motion, $\sigma=(\sigma_t)_{t\geq0}$ is an $(\cf_t)$-adapted process such that $\int_0^t \sigma_s^2\,ds < +\infty$ a.s. for all $t\ge 0$, and $s_0>0$ is the initial SPX price. By continuous model we mean that the SPX has no jump, while the volatility process $\sigma$ may be discontinuous.
The local volatility function associated to Model (\ref{eq:continuous_model_on_S}) is the function $\sigma_{\mathrm{loc}}$ defined by
\bea
\sigma_{\mathrm{loc}}^2(t,x) := \EE[\sigma_t^2|S_t=x]. \label{eq:def_sigma_loc}
\eea
The associated local volatility model is defined by:
\beaa
\frac{dS^{\mathrm{loc}}_t}{S^{\mathrm{loc}}_t} = \sigma_{\mathrm{loc}}(t,S_t^{\mathrm{loc}})\,dW_t, \qquad S_0^{\mathrm{loc}}=s_0.
\eeaa
From \cite{gyongy}, the marginal distributions of the processes $(S_t)_{t\ge 0}$ and $(S_t^{\mathrm{loc}})_{t\ge 0}$ agree:
\bea
\forall t\ge 0, \quad S_t^{\mathrm{loc}} \overset{(d)}{=} S_t. \label{eq:same_marginals}
\eea

Let $T\ge 0$. By definition, the (idealized) VIX at time $T$ is the implied volatility of a 30 day log-contract on the SPX index starting at $T$. For continuous models (\ref{eq:continuous_model_on_S}), this translates into
\bea
\mathrm{VIX}^2_T = \EE\left[ \frac{1}{\tau}\int_T^{T+\tau} \sigma_t^2 \, dt \middle| \cf_T \right] = \frac{1}{\tau}\int_T^{T+\tau} \EE\left[ \sigma_t^2\middle| \cf_T \right] dt,
\eea
where $\tau = \frac{30}{365}$ (30 days). In the associated local volatility model, since by the Markov property of $(S_t^{\mathrm{loc}})_{t\ge 0}$, $\EE[\sigma_{\mathrm{loc}}^2(t,S_t^{\mathrm{loc}})|\cf_T] = \EE[\sigma_{\mathrm{loc}}^2(t,S_t^{\mathrm{loc}})|S_T^{\mathrm{loc}}]$, the VIX, denoted by $\mathrm{VIX}_{\mathrm{loc},T}$, satisfies
\bea
\mathrm{VIX}_{\mathrm{loc},T}^2 = \frac{1}{\tau}\int_T^{T+\tau} \EE[\sigma_{\mathrm{loc}}^2(t,S_t^{\mathrm{loc}})|S_T^{\mathrm{loc}}] \,dt = \EE\left[\frac{1}{\tau}\int_T^{T+\tau} \sigma_{\mathrm{loc}}^2(t,S_t^{\mathrm{loc}}) \,dt \middle|S_T^{\mathrm{loc}}\right]. \label{eq:VIX2loc}
\eea
Note that $\mathrm{VIX}^2_T$ and $\mathrm{VIX}_{\mathrm{loc},T}^2$ have the same mean:
\bea
\EE\left[\mathrm{VIX}^2_T\right] = \EE\left[\mathrm{VIX}_{\mathrm{loc},T}^2\right] = \EE\left[ \frac{1}{\tau}\int_T^{T+\tau} \sigma_t^2 \, dt \right]. \label{eq:same_mean}
\eea

It has often been stated that, within the class of continuous stochastic volatility models calibrated to vanillas, the price of a VIX future is maximized by Dupire's local volatility model. For example, in a general discussion in the introduction of \cite{phl-demarco} about the difficulty of jointly calibrating a stochastic volatility model to both SPX and VIX smiles, De Marco and Henry-Labord\`ere approximate the VIX by the instantaneous volatility, i.e., $\mathrm{VIX}_T \approx \sigma_T$ and $\mathrm{VIX}_{\mathrm{loc},T} \approx \sigma_{\mathrm{loc}}(T,S_T^{\mathrm{loc}})$, and, using Jensen's inequality and (\ref{eq:same_marginals}), they conclude that ``within [the] class of continuous stochastic volatility models calibrated to vanillas, the VIX future is bounded from above by the Dupire local volatility model'':
\begin{multline*}
\EE[\mathrm{VIX}_T] \approx \EE[\sigma_T] = \EE\left[\sqrt{\sigma_T^2}\right] = \EE\left[\EE\left[\sqrt{\sigma_T^2}\middle|S_T\right]\right] \\ \le \EE\left[\sqrt{\EE\left[\sigma_T^2\middle|S_T\right]}\right] = \EE\left[\sigma_{\mathrm{loc}}(T,S_T)\right] = \EE\left[\sigma_{\mathrm{loc}}(T,S_T^{\mathrm{loc}})\right] \approx \EE\left[\mathrm{VIX}_{\mathrm{loc},T}\right].
\end{multline*}
Similarly, one would conclude that within the class of continuous stochastic volatility models calibrated to vanillas, the price of convex options on the \emph{squared} VIX is \emph{minimized} by the local volatility model: for any convex function $f$, such as the call or put payoff function,
\begin{multline}
\EE\left[f\left(\mathrm{VIX}_T^2\right)\right] \approx \EE\left[f\left(\sigma_T^2\right)\right] = \EE\left[\EE\left[f\left(\sigma_T^2\right)\middle|S_T\right]\right] \\ \ge \EE\left[f\left(\EE\left[\sigma_T^2\middle|S_T\right]\right)\right] = \EE\left[f\left(\sigma_{\mathrm{loc}}^2(T,S_T)\right)\right] = \EE\left[f\left(\sigma_{\mathrm{loc}}^2(T,S_T^{\mathrm{loc}})\right)\right] \approx \EE\left[f\left(\mathrm{VIX}_{\mathrm{loc},T}^2\right)\right].
\label{eq:PHL_DeMarco_f}
\end{multline}
(The (correct) fact that $\EE\left[f\left(\sigma_T^2\right)\right] \ge \EE\left[f\left(\sigma_{\mathrm{loc}}^2(T,S_T^{\mathrm{loc}})\right)\right]$ had already been noticed by Dupire in \cite{dupire2005-jensen}.)

In this article, we prove that these statements are in fact incorrect. Even if 30 days is a relatively short horizon, it cannot be harmlessly ignored. VIX are implied volatilities (of SPX options maturing 30 days later), not instantaneous volatilities. We can actually build continuous stochastic volatility models, i.e., processes $(\sigma_t)_{t\ge 0}$, such that
\bea
\EE[\mathrm{VIX}_T] > \EE\left[\mathrm{VIX}_{\mathrm{loc},T}\right] \label{eq:VIXfuture_vs_VIXfutureloc}
\eea
and, more generally, such that for any strictly convex function $f$,
\bea
\EE\left[f\left(\mathrm{VIX}_T^2\right)\right] < \EE\left[f\left(\mathrm{VIX}_{\mathrm{loc},T}^2\right)\right]. \label{eq:inversion_cvx_ordering}
\eea
(Our counterexample actually works for any $\tau>0$.) Not only do we find one convex function $f$ such that (\ref{eq:inversion_cvx_ordering}) holds, we actually build a model in which  (\ref{eq:inversion_cvx_ordering}) holds for any strictly convex function $f$. Actually, we prove an \emph{inversion of convex ordering}: Despite the fact that $\sigma^2_{\mathrm{loc}}(t,S_t^{\mathrm{loc}})$ is smaller than $\sigma_t^2$ in convex order for all $t\in[T,T+\tau]$ (see (\ref{eq:PHL_DeMarco_f})), $\mathrm{VIX}_{\mathrm{loc},T}^2$ is strictly \emph{larger} than $\mathrm{VIX}_T^2$ in convex order. Interestingly, Guyon \cite{guyon-nyu2017} has reported that for short maturities $T$, the market exhibits this inversion of convex ordering: the distribution of $\mathrm{VIX}_{\mathrm{loc},T}^2$ (computed with the market-implied Dupire local volatility) is strictly \emph{larger} than the distribution of $\mathrm{VIX}_T^2$ (implied from the market prices of VIX options) in convex order.

Guyon \cite{guyon-imperial2018} has shown that, when the spot-vol correlation is large in absolute value, stochastic volatility models with fast enough mean reversion, as well as rough volatility models with small enough Hurst exponent, do exhibit this inversion of convex ordering. When mean reversion increases, the distribution of $\mathrm{VIX}^2_T$ becomes more peaked, while the local volatility flattens but not as fast, and as a result it can be numerically checked that $\mathrm{VIX}^2_T$ is strictly smaller than $\mathrm{VIX}_{\mathrm{loc},T}^2$ in convex order for short maturities. Interestingly enough, these models reproduce another characteristic of the SPX/VIX markets: that for larger maturities (say, 3--5 months) the two distributions become non-rankable in convex order, and for even larger maturities $\mathrm{VIX}^2_T$ seems to become strictly larger than $\mathrm{VIX}_{\mathrm{loc},T}^2$ in convex order, i.e., the inversion of convex ordering vanishes as $T$ increases.
However, it is very difficult to mathematically prove the inversion of convex ordering in these models. In order to get a proof of this inversion, our idea is to choose a more extreme model, in which the volatility process $\sigma$ is such that $(\sigma_t)_{t\in[T,T+\tau]}$ is independent of $\cf_T$, so that $\mathrm{VIX}^2_T$ is almost surely constant, but also such that $\mathrm{VIX}_{\mathrm{loc},T}^2$ is not a.s.~constant. In this case, since these two random variables have the same mean (recall (\ref{eq:same_mean})), $\mathrm{VIX}^2_T$ is strictly smaller than $\mathrm{VIX}_{\mathrm{loc},T}^2$ in convex order, and (\ref{eq:inversion_cvx_ordering}) holds for any strictly convex function $f$. In particular, applying (\ref{eq:inversion_cvx_ordering}) with $f(y)=-\sqrt{y}$ yields (\ref{eq:VIXfuture_vs_VIXfutureloc}).

Clearly, in order for $\mathrm{VIX}_{\mathrm{loc},T}^2$ to be non-constant, the local volatility 
cannot be constant as a function of $S$, $\mathrm{d}t$-a.e. in $[T,T+\tau]$.
There are many ways to achieve this, e.g., through volatility of volatility, and it is easy to numerically verify that $\mathrm{VIX}_{\mathrm{loc},T}^2$ (estimated from (\ref{eq:VIX2loc}), e.g., using kernel regressions) is non-constant. However, the main mathematical difficulty here is to \emph{prove} this result. To this end, we will consider models where the non-constant local volatility can be derived in closed form.

The remainder of the article is structured as follows. In Section \ref{sec:simple} we derive a simple counterexample inspired by \cite{bfs} where Beiglb\"ock, Friz, and Sturm use a similar model to prove that local volatility does not minimize the price of options on realized variance. Then we generalize the counterexample in Section \ref{sec:generalization}. Eventually in Section \ref{sec:term-structure} we explain how the model can be extended so that, as suggested by market data, the convex ordering is preserved for long maturities.

\section{A simple counterexample}\label{sec:simple}

Inspired by \cite{bfs}, we fix $T>0$ and consider the following volatility process:
\bea
\sigma_t = \begin{cases}
\sigma_0 & \text{if } t<T+\frac{\tau}{2} \\
\overline{\sigma} & \text{if } t\ge T+\frac{\tau}{2} \text{ and } U=1\quad\;, \\
\underline{\sigma} & \text{if } t\ge T+\frac{\tau}{2} \text{ and } U=-1 
\end{cases} \label{eq:counterexample}
\eea
where $\underline{\sigma}<\sigma_0<\overline{\sigma}$ are three positive constants and $U$ denotes the result of a fair coin toss, independent of $\cf_T$ (e.g., known only at a time $t\in (T,T+\frac{\tau}{2}]$).

\begin{prop}
The stochastic volatility model in \eqref{eq:continuous_model_on_S}, with volatility process described in \eqref{eq:counterexample}, satisfies \eqref{eq:inversion_cvx_ordering}. In particular, VIX futures are strictly more expensive than in their associated local volatility model.
\end{prop}

\begin{proof}
Let us denote 
\beaa
\sigma_+(t) = \begin{cases}
\sigma_0 & \text{if } t<T+\frac{\tau}{2} \\
\overline{\sigma} & \text{if } t\ge T+\frac{\tau}{2} 
\end{cases}, \qquad
\sigma_-(t) = \begin{cases}
\sigma_0 & \text{if } t<T+\frac{\tau}{2} \\
\underline{\sigma} & \text{if } t\ge T+\frac{\tau}{2} 
\end{cases}, 
\eeaa
so that $\sigma_t$ is given by $\sigma_+(t)$ or $\sigma_-(t)$ depending on the coin toss $U$. Since $(\sigma_t)_{t\in[T,T+\tau]}$ is independent of $\cf_T$, $\mathrm{VIX}^2_T$ is a.s. constant:
\beaa
\mathrm{VIX}^2_T = \EE\left[ \frac{1}{\tau}\int_T^{T+\tau} \sigma_t^2 \, dt \right] = \frac{1}{2}\left(\sigma_0^2 + \frac{\overline{\sigma}^2+\underline{\sigma}^2}{2} \right).
\eeaa
Since this is also the mean of $\mathrm{VIX}_{\mathrm{loc},T}^2$, to prove  (\ref{eq:inversion_cvx_ordering}), it is enough to prove that $\mathrm{VIX}_{\mathrm{loc},T}^2$ is not a.s. constant.

Due to the very simple form of Model (\ref{eq:counterexample}), we know the local volatility in closed form:
\bea
\sigma_\text{loc}^2(t,x) = \frac{p_+(t,x)\sigma_+^2(t)+p_-(t,x)\sigma_-^2(t)}{p_+(t,x)+p_-(t,x)},
\label{eq:locvol_counterexample}
\eea
where $p_\pm(t,\cdot)$ is the density of the process $(S^\pm_t)_{t\ge 0}$ with dynamics $\frac{dS^\pm_t}{S^\pm_t} = \sigma_\pm(t) \, dW_t$, $S_0 = s_0$, i.e.,
\beaa
p_\pm(t,x) = \frac{1}{x\sqrt{2\pi\Sigma_\pm(t)}}\exp\left(-\frac{1}{2}\left( \frac{\ln \frac{x}{s_0}}{\sqrt{\Sigma_\pm(t)}} +\frac{1}{2}\sqrt{\Sigma_\pm(t)}\right)^2 \right), \qquad \Sigma_\pm(t) = \int_0^t \sigma_\pm^2(s)\, ds.
\eeaa
Figure \ref{fig:locvol2_counterexample} shows the shape of $\sigma_\text{loc}^2$. Note in particular that $\sigma_\text{loc}$ takes values in $(\underline{\sigma},\overline{\sigma})$ and that 
\bea
\forall t\in\left(T+\frac{\tau}{2},T\right], \qquad \lim_{x\rightarrow +\infty} \sigma_\text{loc}(t,x) = \overline{\sigma}.
\label{eq:sigma_loc_limit}
\eea

\begin{figure}
\begin{center}
\includegraphics[width=10cm,height=8cm]{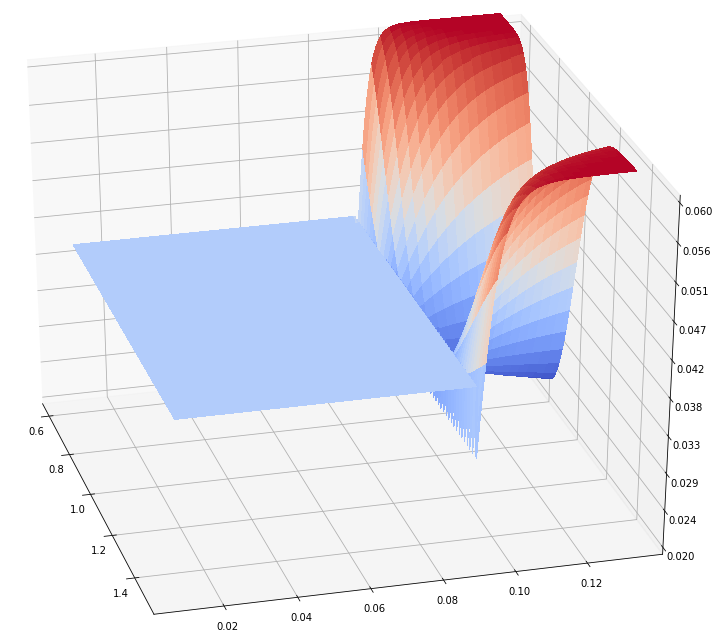}
\end{center}
\caption{Graph of $(t,x)\mapsto \sigma_\text{loc}^2(t,x)$ for $T=0.05$, $\sigma_0=0.2$, $\overline{\sigma}=0.25$, $\underline{\sigma}=0.02$, $s_0=1$}
\label{fig:locvol2_counterexample}
\end{figure}

Let us define
\bea
\psi(x) := \EE\left[\frac{1}{\tau}\int_T^{T+\tau} \sigma_{\mathrm{loc}}^2(t,S_t^{\mathrm{loc}}) \,dt \middle|S_T^{\mathrm{loc}}=x\right]
\label{eq:psi}
\eea
so that $\mathrm{VIX}_{\mathrm{loc},T}^2 = \psi\left(S_T^{\mathrm{loc}}\right)$. Note that 
\beaa
\forall x>0, \qquad \psi(x)< \ell := \frac{1}{2}\left(\sigma_0^2 + \overline{\sigma}^2 \right).
\eeaa
Since $S_T^{\mathrm{loc}}$ has support $\RR_+$, in order to prove that $\mathrm{VIX}_{\mathrm{loc},T}^2$ is not a.s. constant, it is enough to prove that $\psi$ tends to $\ell$ when $x$ tends to $+\infty$. This follows from the next lemma.
\end{proof}

\begin{lem}\label{prop:simple}
In Model (\ref{eq:counterexample}), the function $\psi$ defined by (\ref{eq:psi}) satisfies
\beaa
\lim_{x\rightarrow +\infty} \psi(x) = \ell.
\eeaa
\end{lem}

\begin{proof}
Note that $\psi(x) = \frac{1}{2} \left( \sigma_0^2 + \varphi(x)\right)$, where
\beaa
\varphi(x) := \EE\left[\frac{2}{\tau}\int_{T+\frac{\tau}{2}}^{T+\tau} \sigma_{\mathrm{loc}}^2(t,S_t^{\mathrm{loc}}) \,dt \middle|S_T^{\mathrm{loc}}=x\right],
\eeaa
so it is enough to prove that $\varphi(x)$ tends to $\overline{\sigma}^2$ when $x$ tends to $+\infty$. Let $\eps > 0$ and $\eps' := \left( 1 + 2\left(\overline{\sigma}^2 - \underline{\sigma}^2\right) \right)^{-1}\eps$. Let us denote $L_t := \ln(S^\text{loc}_t)$, whose dynamics is given by
\beaa
dL_t = -\frac{1}{2}\sigma_\text{loc}^2\left(t,e^{L_t}\right)\, dt + \sigma_\text{loc}\left(t,e^{L_t}\right)\, dW_t, \qquad L_0 = \ln s_0.
\eeaa
Since $\sigma_\text{loc}$ is bounded, it is easily checked that $c:=\sup_{t\in [T,T+\tau],x\in\RR}\EE[(L_t-L_T )^2|L_T=x]<+\infty$. Let $\Delta:= \sqrt{\frac{c}{\eps'}}$. Then
\bea
\forall t\in [T,T+\tau], \;\; \forall x\in\RR, \quad \PP(|L_t-L_T|\ge \Delta|L_T=x) \le \frac{\EE[(L_t-L_T )^2|L_T=x]}{\Delta^2}\le \frac{c}{\Delta^2} = \eps'. \label{eq:Delta}
\eea
We have
\beaa
\overline{\sigma}^2 - \varphi(e^x) = \frac{2}{\tau}\int_{T+\frac{\tau}{2}}^{T+\tau}\EE\left[ \left(\overline{\sigma}^2 - \sigma_{\mathrm{loc}}^2(t,e^{L_t})\right)\middle|L_T=x\right]dt  = I_1(x) + I_2(x) + I_3(x),
\eeaa
where
\beaa
I_1(x) &:=& \frac{2}{\tau}\int_{T+\frac{\tau}{2}}^{T+\tau} \EE\left[ \left(\overline{\sigma}^2 - \sigma_{\mathrm{loc}}^2(t,e^{L_t}) \ind_{L_t\le L_T-\Delta} \right)\middle|L_T=x\right] dt, \\
I_2(x) &:=& \frac{2}{\tau}\int_{T+\frac{\tau}{2}}^{T+\frac{\tau}{2}(1+\eps')} \EE\left[ \left(\overline{\sigma}^2 - \sigma_{\mathrm{loc}}^2(t,e^{L_t}) \ind_{L_t> L_T-\Delta} \right)\middle|L_T=x\right] dt, \\
I_3(x) &:=& \frac{2}{\tau}\int_{T+\frac{\tau}{2}(1+\eps')}^{T+\tau} \EE\left[ \left(\overline{\sigma}^2 - \sigma_{\mathrm{loc}}^2(t,e^{L_t}) \ind_{L_t> L_T-\Delta} \right)\middle|L_T=x\right] dt.
\eeaa
Recall that $\sigma_\text{loc}$ takes values in $(\underline{\sigma},\overline{\sigma})$. From (\ref{eq:Delta}), $0\le I_1(x)\le \left(\overline{\sigma}^2 - \underline{\sigma}^2\right)\eps'$ for all $x\in\RR$. Obviously, $0\le I_2(x)\le \left(\overline{\sigma}^2 - \underline{\sigma}^2\right)\eps'$ for all $x\in\RR$. Moreover, it is easy to check that the convergence (\ref{eq:sigma_loc_limit}) is uniform w.r.t. $t\in[T+\frac{\tau}{2}(1+\eps'),T+\tau]$: there exists $x^*$ such that
\beaa
\forall t\in\left[T+\frac{\tau}{2}(1+\eps'),T+\tau\right], \quad \forall x\ge x^*, \qquad 0\le \overline{\sigma}^2 - \sigma_{\mathrm{loc}}^2(t,e^x) \le \eps'.
\eeaa
As a consequence, for all $x\ge x^*+\Delta$, $0\le I_3(x) \le \eps'$. Finally,
\beaa
\forall x\ge x^*+\Delta, \qquad 0\le \overline{\sigma}^2 - \varphi(e^x) \le \left( 1 + 2 \left(\overline{\sigma}^2 - \underline{\sigma}^2\right) \right)\eps' = \eps.
\eeaa
We have thus proved that $\varphi(e^x)$, hence $\varphi(x)$, tends to $\overline{\sigma}^2$ when $x$ tends to $+\infty$.
\end{proof}

\begin{rem}
Note that, if we fix $t_1\in(0,\tau)$ and define
\beaa
\sigma_t = \begin{cases}
\sigma_0 & \text{if } t<t_1 \\
\overline{\sigma} & \text{if } t\ge t_1 \text{ and } U=1\quad\;,\\
\underline{\sigma} & \text{if } t\ge t_1 \text{ and } U=-1 
\end{cases}
\eeaa
{with $U$ only known at time $t_1$,} then we have built a model where the inversion of convex ordering holds for every short maturity $T\in(0,t_1)$.
\end{rem}

\section{Generalization}\label{sec:generalization}
In this section, we generalize the example presented in Section~\ref{sec:simple}, to show that the desired inversion of convex ordering can be obtained with a more interesting structure for the volatility. 
We fix $0<t_1<\tau<t_2:=t_1+\tau$, and define a c\`adl\`ag process $\sigma$ on $[0,t_2)$, which is independent of $\cf^W$, the filtration generated by $W$. 
We start by setting $\sigma_t$ constant equal to $\sigma_0>0$ for $t\in[0,t_1)$. This ensures that $\cf_T=\cf^{S,W}_T=\cf^{W}_T$, and as a consequence $(\sigma_t)_{t\in[T,T+\tau]}$ is independent of $\cf_T$ for all $0<T< t_1$, thus $\mathrm{VIX}_T^2$ is constant:
\[
\mathrm{VIX}_T^2=\EE\left[\frac{1}{\tau}\int_T^{T+\tau}\sigma_t^2 dt\right].
\] 
We shall now define $\sigma$ in $[t_1,t_2)$, with the aim of having $\mathrm{VIX}_{\mathrm{loc},T}^2$ not constant. 
Let $0<\underline{v}\leq\sigma_0\leq\overline{v}$, and $(\sigma_t)_{t\in[t_1,t_2)}$ take values in $[\underline{v},\overline{v}]$.
We denote by $\Lambda$ the law of $(\sigma_t)_{t\in[0,t_2)}$ on $\cd=\cd[0,t_2)$, the space of c\`adl\`ag functions on $[0,t_2)$. 
Note that, for $\Lambda=\frac12(\delta_{\sigma_+}+\delta_{\sigma_-})$ and $T=t_1-\tau/2$, we recover the example of Section \ref{sec:simple}.

For every path $g\in\cd$, we denote by $S^g$ the evolution of the stock price for this realization of $\sigma$, that is
\[
\frac{dS^g_t}{S^g_t}=g(t) \,dW_t,\quad 0\leq t<t_2,\quad S_0^g=s_0,
\]
and by $p_g(t,.)$ the density of the process $S^g_t$, that is
\[
p_g(t,x) = \frac{1}{x\sqrt{2\pi\Sigma_g(t)}}\exp\left(-\frac{1}{2}\left( \frac{\ln \frac{x}{s_0}}{\sqrt{\Sigma_g(t)}} +\frac{1}{2}\sqrt{\Sigma_g(t)}\right)^2 \right), \qquad \Sigma_g(t) = \int_0^t g(s)^2\, ds.
\]
The local volatility then takes the form
\begin{equation}\label{sigma_loc_f}
\sigma_\text{loc}^2(t,x) =\int_{\cd} g(t)^2 q_g(t,x) \,d\Lambda(g),\quad t\in[0,t_2),
\end{equation}
where
\[
q_g(t,x)=\frac{p_g(t,x)}{\int_{\cd} p_h(t,x) \,d\Lambda(h)}.
\]

\begin{lem}
For all $t\ge 0$, the following limit holds for the local volatility:
\begin{equation}\label{eq_lim_sigma_f}
\lim_{x\to+\infty}\sigma_\mathrm{loc}^2(t,x) =\frac{1}{\Lambda(A^{[t]})}\int_{A^{[t]}} g(t)^2 \,d\Lambda(g)=:\overline{\sigma}(t)^2,
\end{equation}
where
\begin{equation}\label{eq:At}
A^{[t]}:=\{g\in\cd : \Sigma_g(t)=\Lambda\text{-}\esssup_{h\in\cd}\Sigma_h(t)\}.
\end{equation}
\end{lem}

\begin{proof}
To study the limit of $\sigma_\text{loc}^2(t,x)$ for $x\to+\infty$, thanks to \eqref{sigma_loc_f} and dominated convergence, we are reduced to consider the limit of $q_g(t,x)$.
Note that
\begin{equation}\label{eq_qf}
q_g(t,x)^{-1}=\int_{\cd}F(g,h,t,x)\,d\Lambda(h),
\end{equation}
where
\[
\textstyle{F(g,h,t,x):= \sqrt{\frac{\Sigma_g(t)}{\Sigma_h(t)}} \exp\left\{-\frac{1}{2}\left[ 
\left(\ln \frac{x}{s_0}\right)^2\left(\frac{1}{\Sigma_h(t)}-\frac{1}{\Sigma_g(t)}\right) +\frac{1}{4}(\Sigma_h(t)-\Sigma_g(t))
\right]\right\}.}
\]
By Fatou's lemma, $\lim_{x\to+\infty}\int_{\cd}F(g,h,t,x)\,d\Lambda(h)=+\infty$ as soon as $\Lambda(\cd_{g,t})>0$,
where
\[
\cd_{g,t}:=\{h\in\cd : \Sigma_h(t)>\Sigma_g(t)\},
\]
which in turn implies $\lim_{x\to+\infty}q_g(t,x)=0$. 
On the other hand, if $\Lambda(\cd_{g,t})=0$, then by dominated convergence $\lim_{x\to+\infty}\int_{\cd}F(g,h,t,x)\,d\Lambda(h)=\int_{\cd}\lim_{x\to+\infty}F(g,h,t,x)\,d\Lambda(h)$, being $\sigma$ bounded and bounded away from zero. 
Now $\lim_{x\to+\infty}F(g,h,t,x)$ equals zero when $\Sigma_h(t)<\Sigma_g(t)$, and one when $\Sigma_h(t)=\Sigma_g(t)$.
This concludes the proof, noticing that
\[
A^{[t]}=\{g\in\cd : \Lambda(\cd_{g,t})=0\}.
\]
\end{proof}

\begin{prop}\label{prop:gen}
Consider the stochastic volatility model \eqref{eq:continuous_model_on_S}. Let $\sigma$ be constant equal to $\sigma_0>0$ in $[0,t_1)$, independent of $\cf^W$ in $[t_1,t_2)$, admitting only finitely many paths, so that
\begin{equation}\label{eq lambda discrete}
\Lambda=\sum_{n=1}^N u_n\delta_{g_n},\quad \text{with}\; N\ge 2, \; g_n\in\cd, \; u_n > 0, \; \textstyle{\sum_{n=1}^Nu_n=1},
\end{equation}
with $g_n(t)=\sigma_0$ for all $n\in\{1,\ldots,N\}$, $t\in[0,t_1)$. We assume the following non-degeneracy condition of $\sigma$ in a neighborhood of $t_1$: There exist $m,n\in\{1,\ldots,N\}$ and $\varepsilon>0$ such that $g_m(t) \neq g_n(t)$ for all $t \in [t_1,t_1+\varepsilon]$. Then, for all maturities $T<t_1$, \eqref{eq:inversion_cvx_ordering} holds. In particular, VIX futures are strictly more expensive than in the associated local volatility model.
\end{prop}

With an abuse of notation, below we will write $q_n, \Sigma_n, F(n,m,t,x)$ instead of $q_{g_n}, \Sigma_{g_n}, F(g_n,g_m,t,x)$, to ease readability.

\begin{proof}
As in the example of Section \ref{sec:simple}, we consider the function
\[
\psi(x) := \EE\left[\frac{1}{\tau}\int_T^{T+\tau} \sigma_{\mathrm{loc}}^2(t,S_t^{\mathrm{loc}}) \,dt \middle|S_T^{\mathrm{loc}}=x\right]
\]
and note that our non-degeneracy assumption implies that
\[
\forall x>0,\quad \psi(x) < \frac{1}{\tau}\int_T^{T+\tau} 
 \overline{\sigma}(t)^2\,dt=:\ell.
\]
To prove the inversion of convex ordering for all maturities $T<t_1$, we will show that $\lim_{x\to+\infty}\psi(x)=\ell$, that is,
\begin{equation}\label{eq liml}
\int_{t_1}^{T+\tau}\EE[\sigma_{\mathrm{loc}}^2(t,S_t^{\mathrm{loc}}) | S_T^{\mathrm{loc}}=x]\ dt \;\;
\xrightarrow[x \to +\infty]{}\;\;
\int_{t_1}^{T+\tau}\overline{\sigma}(t)^2\ dt.
\end{equation}

Since $\Lambda$ is discrete and the functions $t\mapsto\Sigma_n(t)$ are continuous and bounded in $[t_1,t_2)$, this interval divides in countably many intervals $I_k=[a_k,b_k)$, $k\in\NN$, $a_k<b_k$, such that, in each open interval $(a_k,b_k)$, the function $\overline{\sigma}$ defined in \eqref{eq_lim_sigma_f} coincides with one or more paths of $\sigma$. To be more precise, for every $k\in\NN$, the sets $A^{[t]}$ defined in \eqref{eq:At} coincide for every $t\in(a_k,b_k)$, say to a set $A^k$, and
\begin{equation}\label{eq sn}
\overline{\sigma}(t)=g_n(t)\;\; \text{for}\; t\in(a_k,b_k),\quad \text{for all}\; g_n\in A^k.
\end{equation}

To show the convergence in \eqref{eq liml}, we split the interval $[t_1,T+\tau]$ into subintervals $\widetilde I_k:=[t_1,T+\tau]\cap I_k$, $k\in\NN$. Since by dominated convergence
\begin{equation*}
\lim_{x\to+\infty}\sum_{k\in\NN}\int_{\widetilde I_k}\EE[\overline{\sigma}(t)^2-\sigma_{\mathrm{loc}}^2(t,S_t^{\mathrm{loc}}) | S_T^{\mathrm{loc}}=x]\ dt=\sum_{k\in\NN}\lim_{x\to+\infty}\int_{\widetilde I_k}\EE[\overline{\sigma}(t)^2-\sigma_{\mathrm{loc}}^2(t,S_t^{\mathrm{loc}}) | S_T^{\mathrm{loc}}=x]\ dt,
\end{equation*}
we are reduced to prove that for all $k\in\NN$,
\begin{equation}\label{eq inv}
\lim_{x\to+\infty}\int_{\widetilde I_k}\EE[\overline{\sigma}(t)^2-\sigma_{\mathrm{loc}}^2(t,S_t^{\mathrm{loc}}) | S_T^{\mathrm{loc}}=x]\ dt=0.
\end{equation}

Fix $k\in\NN$ and $\eps_k>0$, and set $\varepsilon_k':=\min\{\varepsilon_k(b_k-a_k+3(\bar{v}^2-\underline{v}^2))^{-1},(b_k-a_k)/3\}$. We split the interval $I_k$ into three subintervals
\begin{equation}\label{eq subint}
J'_k:=[a_k, a_k+\varepsilon'_k],\quad
J_k:=(a_k+\varepsilon'_k,b_k-\varepsilon'_k),\quad J''_k:=[b_k-\varepsilon'_k,b_k),
\end{equation}
and we are going to show that $\sigma_\text{loc}^2(t,x)$ converges uniformly to $\overline{\sigma}(t)^2$ w.r.t. $t\in J_k$, for $x\to+\infty$.

Let $N^k:=\{n\in\{1,...,N\} : g_n\in A^k\}$ and note that the function $F(n,m,t,x)$ depends on the paths $g_n$ and $g_m$ only through $\Sigma_n(t)$ and $\Sigma_m(t)$, respectively. Therefore, from \eqref{eq_qf}, we have
\begin{equation}\label{eq q1}
\frac{1}{q_n(t,x)}=\sum_{m=1}^N F(n,m,t,x) u_m = F(n,m_k,t,x)\Lambda(A^k)+\sum_{m\not\in N^k} F(n,m,t,x) u_m,\quad t\in J_k,
\end{equation}
for any $m_k\in N^k$, which reduces to $\Lambda(A^k)+\sum_{m\not\in N^k}F(n,m,t,x) u_m$ for $n\in N^k$.
Now, it is easy to verify that, when $x$ tends to $+\infty$, $F(n,m,t,x)$ converges to zero uniformly w.r.t. $t\in J_k$ whenever $n\in N^k$ and $m\not\in N^k$. In particular, there is $x_k$ such that, for all $x\geq x_k$, $t\in J_k$, and $n\in N^k$, $\sum_{m\not\in N^k}F(n,m,t,x) u_m\leq\varepsilon_k'\Lambda(A^k)^2\bar{v}^{-2}$, thus
\begin{equation}\label{eq qin}
\left|q_n(t,x)-\frac{1}{\Lambda(A^k)}\right|=
\frac{\sum_{m\not\in N^k}F(n,m,t,x) u_m}{\Lambda(A^k)(\Lambda(A^k)+\sum_{m\not\in N^k}F(n,m,t,x) u_m)}
\leq\varepsilon_k'\bar{v}^{-2}.
\end{equation}

Since $F(n,m,t,x)=F(m,n,t,x)^{-1}$, we also have that, when $x$ tends to $+\infty$, $F(n,m,t,x)$ converges to $+\infty$ uniformly w.r.t. $t\in J_k$ whenever $n\not\in N^k$ and $m\in N^k$. This gives the existence of $y_k$ such that, for all $x\geq y_k$, $t\in J_k$, $n\not\in N^k$, and $m\in N^k$,
\[
F(n,m,t,x)\geq\bar{v}^2(\varepsilon_k'\Lambda(A^k))^{-1},
\]
which by \eqref{eq q1} implies
\begin{equation}\label{eq qnotin}
q_n(t,x)\leq \varepsilon_k'\bar{v}^{-2}.
\end{equation}
Note that in the present setting we have
\[
\sigma_\text{loc}^2(t,x)=\sum_{n=1}^N u_n g_n(t)^2 q_n(t,x)\qquad\text{and}\qquad
\overline{\sigma}(t)^2=\frac{1}{\Lambda(A^k)}\sum_{n\in N^k}u_ng_n(t)^2, \quad t\in J_k,
\]
from \eqref{eq sn}.
Now \eqref{eq qin} and \eqref{eq qnotin} imply 
\begin{equation}\label{eq unifb}
|\sigma_\text{loc}^2(t,x) -\overline{\sigma}(t)^2|\leq\varepsilon_k',\quad \text{for all $x\geq z_k:=\max\{x_k,y_k\}$ and $t\in J_k$},
\end{equation}
which shows the claimed uniform convergence.

As in the proof of Lemma~\ref{prop:simple}, we consider the log-price process $L_t=\ln(S^\text{loc}_t)$, and we have
$c_k:=\sup_{t\in I_k,x\in\RR}\EE[(L_t-L_T )^2|L_T=x]<+\infty$, $k\in\NN$, since $\sigma$ is bounded. Setting $\Delta_k:= \sqrt{{c_k}/{\eps_k'}}$, we again obtain 
\begin{equation}\label{eq cheb}
\PP(|L_t-L_T|\ge \Delta_k|L_T=x) \le \eps_k',\quad \text{for all $t\in I_k$ and $x\in\RR$}.
\end{equation}

We are going to show that
\[
\int_{\widetilde I_k}\left(\overline{\sigma}(t)^2-\EE[\sigma_{\mathrm{loc}}^2(t,S_t^{\mathrm{loc}}) | S_T^{\mathrm{loc}}=e^x]\right) dt = \int_{\widetilde I_k}\EE[\overline{\sigma}(t)^2-\sigma_{\mathrm{loc}}^2(t,e^{L_t}) | L_T=x] \,dt
\]
converges to zero for $x\to+\infty$, by proving that for $x$ big enough this is smaller than the arbitrarily chosen $\eps_k$.
This in turn implies \eqref{eq inv}, being $k\in\NN$ arbitrary, and concludes the proof of \eqref{eq liml}.
In order to do that, we divide $\widetilde I_k$ in three subintervals :
\[
\widetilde{J}_k:=[t_1,T+\tau]\cap J_k,\quad
\widetilde{J}'_k:=[t_1,T+\tau]\cap J'_k,\quad
\widetilde{J}''_k:=[t_1,T+\tau]\cap J''_k,\quad k\in\NN,
\]
where we used the notation introduced in \eqref{eq subint}.
Note that, since $\sigma$ takes values in $[\underline{v},\overline{v}]$, 
\[
\int_{\widetilde{J}'_k}\EE[\overline{\sigma}(t)^2-\sigma_{\mathrm{loc}}^2(t,e^{L_t}) | L_T=x]\ dt\leq (\overline{v}^2-\underline{v}^2)\eps'_k,
\]
and the same bound holds when taking the integral over $\widetilde{J}''_k$.
On the other hand, \eqref{eq cheb} implies
\[
\int_{\widetilde{J}_k}\EE[(\overline{\sigma}(t)^2-\sigma_{\mathrm{loc}}^2(t,e^{L_t}))\ind_{L_t\le L_T-\Delta_k} | L_T=x]\ dt\leq (\overline{v}^2-\underline{v}^2)\eps'_k,
\]
and \eqref{eq unifb} implies
\[
\int_{\widetilde{J}_k}\EE[(\overline{\sigma}(t)^2-\sigma_{\mathrm{loc}}^2(t,e^{L_t}))\ind_{L_t> L_T-\Delta_k} | L_T=x]\ dt\leq \eps'_k|\widetilde{J}_k|,
\]
for all $x\geq \ln(z_k)+\Delta_k$.
Altogether, for $x\geq \ln(z_k)+\Delta_k$ we have
\[
\int_{\widetilde{I}_k}\EE[\overline{\sigma}(t)^2-\sigma_{\mathrm{loc}}^2(t,e^{L_t}) | L_T=x]\ dt\leq (b_k-a_k+3(\overline{v}^2-\underline{v}^2))\eps'_k\leq\eps_k.
\]
This concludes the proof.
\end{proof}

\section{Term-structure of convex ordering}\label{sec:term-structure}
In this section, we extend the model built in Section~\ref{sec:generalization} in order to have the convex ordering preserved for long maturities, as suggested by market data. To this end, we set
\bea \label{eq:SLV}
\frac{dS_t}{S_t}= \sigma_0\frac{Y}{\sqrt{\EE[Y^2 | S_t]}} \, dW_t,\qquad t\geq t_2,
\eea
where $\sigma_0\in\RR_+$ and $Y$ is a Bernoulli random variable known in $t_2$ and independent of anything else. Say $Y$ takes the value $y_-$ with probability $q_-$ and $y_+$ with probability $q_+=1-q_-$, for some $0<y_-<y_+$ and $0<q_-<1$. By Jourdain and Zhou \cite{jourdain}, as long as the ratio $y_+/y_-$ is not too large, the stochastic differential equation (SDE) \eqref{eq:SLV} admits a weak solution $(\Omega,(\cf_t),\PP,W,(S_t),Y)$, which may not be unique. In the following we use the subscript or superscript $\PP$ to emphasize that a priori the corresponding quantities depend on the weak solution of \eqref{eq:SLV}.

Note that (\ref{eq:SLV}) implies that, whatever the weak solution, $\sigma_{\mathrm{loc},\PP}^2(t,S_t)=\sigma_0^2$ for $t\geq t_2$. Therefore $\mathrm{VIX}_{\mathrm{loc},T}^2$ does not depend on the weak solution and is constant equal to $\sigma_0^2$ for all maturities $T\geq t_2$. We now want to show that, on the other hand, for any weak solution of \eqref{eq:SLV}, this is not true for $\mathrm{VIX}_{\PP,T}^2$. This will imply that $\mathrm{VIX}_{\mathrm{loc},T}^2$ is strictly smaller than $\mathrm{VIX}_{\PP,T}^2$ in convex order for $T\geq t_2$, thus there is no inversion of convex ordering for long maturities.

For any weak solution of \eqref{eq:SLV}, we set
\beaa
F_\PP(s,t,x):=\EE^\PP[Y^2|S_{t_2}=s,S_t=x].
\eeaa
Since $Y$ is independent of $W$, the conditional law of $(S_t)_{t\geq t_2}$ given $Y=y_\pm$ and $S_{t_2}=s$ under $\PP$ agrees with the (unique) law of a weak solution to the SDE
\[
\frac{dS_t^{s,\pm,\PP}}{S^{s,\pm,\PP}_t}=\sigma_0\frac{y_\pm}{\sqrt{F_\PP\left(s,t,S^{s,\pm,\PP}_t\right)}} \, d\widetilde{W}_t,\qquad t\geq t_2,\qquad S^{s,\pm,\PP}_{t_2}=s,
\]
living possibly on a different probability space (the weak uniqueness of the solution follows from \cite[Theorem~3]{veretennikov}, given that \cite[Proposition~5.1]{brunick} ensures the existence of a measurable version of $F_\PP(s,\cdot,\cdot)$).
Being $F_\PP$ bounded and bounded away from zero, we deduce that, for $t>t_2$, the conditional law of $S_t$ given $Y=y_\pm$ and $S_{t_2}=s$ under $\PP$ admits a density $p^\PP_\pm(s,t,x)$, and that $p^\PP_\pm(s,t,x)>0$ for all $x\in\RR_+$, which in turn implies that 
\bea \label{eq range F}
F_\PP(s,t,x)=\frac{q_-y^2_-p^\PP_-(s,t,x)+q_+y^2_+p^\PP_+(s,t,x)}{q_-p^\PP_-(s,t,x)+q_+p^\PP_+(s,t,x)}\in(y^2_-,y^2_+),\qquad t> t_2.
\eea
Then, for $T\geq t_2$, we have
\[
\mathrm{VIX}^2_{\PP,T} = \sigma_0^2Y^2\frac{1}{\tau}\int_T^{T+\tau} \EE^\PP\left[ \frac{1}{F_\PP(S_{t_2},t,S_t)}\middle| \cf_T \right] dt=: \sigma_0^2Y^2 \Psi_\PP.
\]
Now, having $\mathrm{VIX}^2_{\PP,T}$ constant (thus necessarily equal to $\sigma_0^2$) corresponds to having $Y^2 \Psi_\PP\equiv 1$, which is not possible given that
$\Psi_\PP$ takes values in $\left(\frac{1}{y_+^2},\frac{1}{y_-^2}\right)$, by \eqref{eq range F}. This shows that  $\mathrm{VIX}^2_{\PP,T}$ cannot be constant for any $T\geq t_2$.

\begin{rem}
To the best of our knowledge, uniqueness of a weak solution of \eqref{eq:SLV} is still an open question. More generally, partial results on the existence of a weak solution of a calibrated stochastic local volatility (SLV) model of the form
\bea
\frac{dS_t}{S_t}= \sigma_\mathrm{Dup}(t,S_t)\frac{f(Y_t)}{\sqrt{\EE[f(Y_t)^2 | S_t]}} \, dW_t 
\label{eq:SLV_general}
\eea
have been obtained in \cite{abergel,jourdain}, but uniqueness has not been addressed. Note that Lacker \emph{et al}.\ \cite{lsz} have recently proved the weak existence and uniqueness of a \emph{stationary} solution of a similar nonlinear SDE with drift, under some conditions. However, their result does not apply to the calibration of SLV models. Indeed, market-implied risk neutral distributions $(\cl(S_t))_{t\ge 0}$ are strictly increasing in convex order and therefore no stationary solution $(S_t,Y_t)_{t\ge 0}$ can be a calibrated SLV model.

The possible absence of uniqueness of a weak solution of \eqref{eq:SLV} or \eqref{eq:SLV_general} is problematic, not only theoretically but also practically. It means that the price of a derivative in the calibrated SLV model may not be well defined. For example, in our case, the VIX may depend on $\PP$. More generally, existence and uniqueness of \eqref{eq:SLV_general} for general processes $(Y_t)_{t\ge 0}$ such as It\^o processes remain a very challenging, open problem, despite the fact that these models are widely used in the financial industry, in particular thanks to the particle method of Guyon and Henry-Labord\`ere \cite{jg-phl-particle}.
\end{rem}

\bigskip

{\bf Acknowledgements.} We would like to thank Bruno Dupire and Vlad Bally for interesting discussions and helpful comments.


\begin{thebibliography}{1}


\bibitem{abergel} Abergel, F., Tachet, R. : \emph{A nonlinear partial
integrodifferential equation from mathematical finance}, Discrete Cont.
Dynamical Systems, Serie A, 27(3):907--917, 2010.

\bibitem{bfs}Beiglb\"ock, M., Friz, P., Sturm, S.: \emph{Is the minimum value of an option on variance generated by local volatility?}, SIAM J. Finan. Math. 2:213--220, 2011. 

%
%
%
%
%
%
\bibitem{brunick} Brunick, G., Shreve, S.: \emph{Mimicking an It{\^o} process by a solution of a stochastic differential equation}, The Annals of Applied Probability, 23(4):1584--1628, 2013.
%
%
%
%

\bibitem{phl-demarco}De Marco, S., Henry-Labord\`ere, P.: \emph{Linking
vanillas and VIX options: A constrained martingale optimal transport
problem}, SIAM J. Finan. Math., 6:1171--1194, 2015.

%
\bibitem{dupire}Dupire, B.: \emph{Pricing with a smile}, Risk, January,
1994.

\bibitem{dupire2005-jensen}Dupire, B.: \emph{Exploring Volatility Derivatives: New Advances in Modelling}, presentation at Global Derivatives, Paris, 2005.
%
%


%
%
%
%
%

\bibitem{jg-phl-particle} Guyon, J., Henry-Labord\`ere, P., \emph{Being Particular About Calibration}, Risk, January, 2012.
Long version \emph{The smile calibration problem solved}, SSRN preprint available at {\tt ssrn.com/abstract=1885032}, 2011.

\bibitem{guyon-nyu2017}Guyon, J: \emph{On the Joint Calibration of SPX and VIX Options}, presentation at Jim Gatheral's 60th Birthday Conference, NYU, October 14, 2017.

\bibitem{guyon-imperial2018}Guyon, J: \emph{On the Joint Calibration of SPX and VIX Options}, presentation at the Finance and Stochastics seminar, Imperial College, London, March 28, 2018.

\bibitem{gyongy}Gyöngy, I.: \emph{Mimicking the One-Dimensional Marginal
Distributions of Processes Having an Itô Differential}, Probability
Theory and Related Fields, 71, 501-516, 1986.
%
%
\bibitem{jourdain} Jourdain, B., Zhou, A.: \emph{Fake Brownian motion and calibration of a Regime Switching Local
Volatility model}, preprint, available at {\tt arxiv.org/abs/1607.00077v1}, 2016.

\bibitem{lsz} Lacker, D., Shkolnikov, M., Zhang, J.: \emph{Inverting the Markovian projection, with an application to local stochastic volatility models}, preprint, 2019.

%
%
%
%
%
%


\bibitem{veretennikov} Veretennikov, A Yu: \emph{On the strong solutions of stochastic differential equations}, Theory of Probability \& Its Applications, 24(2):354--366, 1980.


\end{thebibliography}
\end{document}